\documentclass[english,draft]{svjour3}
\usepackage{amsfonts}
\usepackage{amsmath}
\usepackage{geometry}
\usepackage{amssymb}
\usepackage{graphicx}
\usepackage{color}
\usepackage[T1]{fontenc}
\usepackage[latin9]{inputenc}

\makeatletter
\numberwithin{equation}{section}
\numberwithin{figure}{section}
\DeclareMathOperator*{\argmax}{arg-max}
\newcommand{\musiclab}{\textsc{MusicLab}}

\begin{document}

\title{Optimizing Expected Profit in a Multinomial Logit Model \\ with Position Bias and Social Influence}

\author{
A.~Abeluik 
\and
G.~Berbeglia 
\and
M.~Cebrian 
\and
P.~Van~Hentenryck
}
\institute{
A.~Abeliuk and M.~Cebrian \at The University of Melbourne and National ICT Australia.
\and
G.~Berbeglia \at Melbourne Business School (The University of Melbourne) and National ICT Australia. 
\and
P.~Van~Hentenryck \at The Australian National University and National ICT Australia. \email{pvh@nicta.com.au}.
}



\date{}

\maketitle

\begin{abstract}
  Motivated by applications in retail, online advertising, and
  cultural markets, this paper studies how to find the optimal
  assortment and positioning of products subject to a capacity
  constraint and social influence. We prove that the optimal
  assortment and positioning can be found in polynomial time for a
  multinomial logit model capturing utilities, position bias, and
  social influence.  Moreover, in a dynamic market, we show that the
  policy that applies the optimal assortment and positioning and
  leverages social influence outperforms in expectation any policy not
  using social influence.
\end{abstract}

\section{Introduction}

One of the most studied problems in the area of revenue management is
the optimal assortment problem. Informally, the problem consists in
selecting a subset of products to offer to consumers so that the
expected profit is maximised. Such optimal subset depends on the
profits obtained by selling a unit of each of these products, as well
as the purchasing behaviour of the consumers, which is represented
using a discrete choice model. Most of the discrete choice models of
practical significance are special cases of the Random Utility Model
(RUM) in which each product $i$ is characterized by a distribution
$D_i$ that represents the consumer utility. When offered a subset $S$
of products, each consumer can then be viewed as a realization of all
the distributions $D_i$ and she picks the product in $S$ with the
highest utility.  Some of the most studied discrete choice models are
the Multinomial Logit (MNL) \cite{luce1959}, the \emph{nested
  multinomial logit} (NMNL) \cite{williams1977formation} and the Mixed
Multinomial Logit (MMNL) \cite{daly1978improved}.  When consumers
choose products according to the MNL model, the optimal assortment
problem can be solved efficiently \cite{talluri2004revenue}. This
result was later extended in \cite{rusmevichientong2010dynamic} to the
case in which there is a constraint on the maximum number of products
that can be offered. However, when consumers follow either an NMNL or
an MMNL model, the optimal assortment problem is NP-hard
\cite{davis2014assortment,rusmevichientong2010assortment,bront2009column}.

The underlying assumption in the vast literature that studies the
optimal assortment problem is that the consumer choice behaviour is
solely affected by the subset of products being offered: The
particular way in which the products are displayed has no importance.
This assumption, however, is violated in many real-world situations.
Moreover, in many settings, social signals (e.g., quality ratings and
recommendations) are displayed together with the products. Once again,
this is traditionally ignored in the literature. 

The problem of finding an assortment and positioning of products that
maximizes profit subject to capacity constraints and social influence
has numerous practical applications, from retail to online advertising
and cultural markets. Position bias is pervasive in E-Commerce and
recommendation systems (e.g.,
\cite{kempe2008cascade,Lerman2014,maille2012sponsored}), while many
settings also display of social signals to influence consumers. Social
science studies often feature both position bias and social influence
\cite{krumme2012quantifying,salganik2006experimental}. It is thus
important for practical applications to study a model that jointly
considers position bias, social influence, and limits on the number of
displayed products although, to our knowledge, such a model has not
been analyzed so far.

This paper addresses this gap and studies a multinomial logit model
that captures these three elements. It contains two main technical
contributions:
\begin{enumerate}
\item It shows that optimizing the expected profit for this
  multinomial logit model can be performed in polynomial time. This
  result holds although the traditional regularity assumption is
  violated in this setting.

\item It proves that it is beneficial to use social influence for
  maximizing the expected profit. In particular, the paper proves
  that, in a dynamic market based on the model, the policy that uses
  social influence and applies the optimal assortment and positioning
  at every step outperforms in expected profit any policy not using
  social influence.
\end{enumerate}

\noindent
The first contribution generalizes the seminal results in
\cite{rusmevichientong2010dynamic} to include position bias and to
study the role of social influence. The second contribution sheds new
light on social influence. Indeed, most studies focus on showing the
negative side of social influence, i.e., the increased
unpredictability and inequalities it creates. This paper shows its
main positive aspect: Its ability to improve the efficiency of the
market.

The rest of the paper is organized as follows. Sections
\ref{section:specification} and \ref{section:related} present the
problem specification and related work. Section
\ref{section:algorithm} proves that maximizing the expected profit in
the proposed model can be performed in polynomial time. Section
\ref{section:social} demonstrates the benefits of social influence. Section
\ref{section:conclusion} concludes the paper.

\section{Problem Specification}
\label{section:specification}

Consider a market with $N$ products ${\cal P} = {1,...,N}$, where only
$c \leq N$ products in ${\cal P}$ can be displayed. In the MNL model,
each product $i \in \mathcal{P}$ has an utility $$u_i = q_i +
\epsilon_i$$ where $q_i$ is a constant representing the inherent
quality of product $i$ and $\epsilon_i$ is a random variable following
a Gumbel distribution with zero mean and representing the error
term. The no-purchase option, denoted by $0$, is added by assuming its
utility is zero, i.e., $e^{q_0}=1$.  Given a subset $S \subseteq
\mathcal{P}$, the probability that a consumer chooses product $i \in
S$ is 
\[
P_i(S) = \frac{e^{q_i}}{\sum_{j \in S}e^{q_j} +  1}.
\]

In this paper, we consider the optimal assortment problem in which
consumers follow the MNL but each product must be displayed in one of
$N$ positions, each of which has a visibility $\theta_i \geq 0$ $(1
\leq i \leq N)$. Without loss of generality, we assume
$\theta_{1}\ge\theta_{2}\geq\cdots\geq\theta_{N}\geq0$. A position
assignment is an injective function $\sigma: S\rightarrowtail 1..N$
that assigns a position to each product in an assortment $S \subseteq
{\cal P}$. The intrinsic utility $u_i$ of product $i$ when displayed in
position $j$ is shifted by a factor $\ln \theta_j$. Since
\[
e^{q_i + \ln \theta_j} = e^{\ln \theta_j} \ e^{q_i} = \theta_j \ e^{q_i},
\]
the probability of selecting product $i$ becomes
\[
P_{i}(S,\sigma)=\frac{\theta_{\sigma_j} u_i }{\sum_{j\in S}\theta_{\sigma_j} u_j +1 }.
\]
where $u_i=e^{q_i}$ for notational simplicity. Note that, for any
assortment $S$ and position assortment $\sigma$, we have
\[
\sum_{i\in S}P_{i}(S,\sigma)+P_{0}(S,\sigma)=1,
\]
where $P_{0}(S,\sigma)$ is the no-choice option.

Let $r_{i}$ denote the marginal profit of product $i$ and assume that
the no-purchase option has no profit, i.e., $r_{0}=0$. The first
problem considered in the paper is to find a polynomial-time
algorithm that maximizes the expected profit in this market, i.e.,
\[
U(S,\sigma)=\sum_{i\in S}r_{i}P_{i}(S,\sigma).
\]
 \begin{definition}[Capacitated Multinomial Logit Assortment and Positioning Problem]
With the above notations, the Capacitated Multinomial Logit Assortment and Positioning Problem (CMLAPP) is defined as 
\[
Z^{*}=\max\left\{ U(S,\sigma) \mid S \subseteq {\cal P} \ \wedge \ |S| \leq c \ \wedge \ \sigma: S \rightarrowtail 1..N \right\}.
\]
\end{definition}

\noindent
This paper also considers a dynamic market where the perceived product
quality $u_i$ varies along time according to a social influence
signal. Consider $d_{i,t}$ the number of selections (e.g., the number
of downloads in a music market or the number of clicks on a website)
of product $i$ at time $t$. We define the utility of product $i$ at
time $t$ as a combination of the inherent quality of the product $u_i$
and a non-decreasing and positive social influence function
$f(\cdot)$, i.e., $$u_{i,t} = u_i + f(d_{i,t}).$$ In this dynamic
market,
\[
d_{i,t} = d_{i,t-1} + 1
\]
if product $i$ is selected at step $t$ and $d_{i,t} = d_{i,t-1}$
otherwise. More formally, the probability of selecting product $i$ for
social signal $d$ given a subset $S$ and a position assignment
$\sigma$ is given by
\[
P_{i}(S,\sigma,d)=\frac{\theta_{\sigma_j} (u_i + f(d_{i}))}{\sum_{j\in
    S} (\theta_{\sigma_j} (u_i + f(d_{i}))) +1}.
\]
In a dynamic market with $T$ steps, the expected profit can be defined
by the following recurrence:
\[
\begin{array}{llll}
U_t(d) & = & \max_{S,\sigma} \sum_{i \in S} P_{i}(S,\sigma,d) (r_i + U_{t+1}(d[i \leftarrow d_i + 1])) & (t \in 1..T) \\
U_{T+1}(d) & = & 0
\end{array}
\]
where $d[i \leftarrow v]$ represents the vector $d$ where element $i$
has been replaced by $v$ and $U_1(\langle 0,\ldots,0\rangle)$ denotes
the overall expected profit.

When there is no social signal, the optimization problem is the same
at each step $t$ and the policy that selects an optimal solution to
the CMLAPP for each $t$ is optimal.  The second problem considered in
this paper is to determine that the effect of social influence on this
policy.

\section{Related Work}
\label{section:related}

This section reviews relevant related work, including some algorithmic
concepts that are important for the results of this paper.

\subsection{The MusicLab Model}

The problem studied in the paper is motivated by the seminal study of
social influence in the {\sc MusicLab} \cite{salganik2006experimental}
and the descriptive model introduced in \cite{krumme2012quantifying}
to capture these experiments. In the {\sc MusicLab}, participants
enter a cultural market and are presented with a number of songs. Only
the title of the song and the band name are displayed. A market
participant chooses a song and, after listening to the song, is given
the opportunity to download it. The {\sc MusicLab} experiments were
designed to measure the impact of social influence in cultural
markets. Participants were split into a number of independent worlds.
In all but one world, the participants were shown a social signal,
i.e., the number of downloads of each song in their world. In the
last world, no social signal was displayed. The experiments were used to
demonstrate the unpredictability of cultural markets in the presence
of social influence.

The descriptive model introduced in \cite{krumme2012quantifying} was
used to capture and reproduce the essence of this cultural market.
The model is defined in terms of a market composed of $n$ songs. Each
song $i \in \{1,\ldots,n\}$ is characterized by two values:
\begin{enumerate}
\item Its {\em appeal} $A_i$ which represents the inherent preference
  of listening to song $i$ based only on its name and its band;

\item Its {\em quality} $q_i$ which represents the conditional
  probability of downloading song $i$ given that it was sampled.
\end{enumerate}
The \musiclab{} experiments present each participant with a playlist
$\pi$, i.e., a permutation of $\{1,\ldots,n\}$. Each position $p$ in
the playlist is characterized by its {\em visibility} $v_p$ which is
the inherent probability of sampling a song in position $p$. The model
specifies the probability of listening to song $i$ at time $k$ given a
playlist $\sigma$ as
\[
p_{i,k}(\sigma) =  \frac{v_{\sigma_i}(\alpha A_i+D_{i,k})}{\sum_{j=1}^n v_{\sigma_j}(\alpha A_j+D_{j,k})},
\]
where $D_{i,k}$ is the number of downloads of song $i$ at time $k$ and
$\alpha > 0$ is a scaling factor which is the same for all
songs. Observe that the probability of sampling a song depends on its
position in the playlist, its appeal, and the number of downloads at
time $k$. 

This paper generalizes the \musiclab{} model in multiple ways: It uses
a well-known multinomial logit model, introduces a limit on how many
products can be displayed (which is more realistic in practice), and
embeds the social influence signal in a non-decreasing, positive
function $f$. All three generalizations are significant for practical
applications.\footnote{The \musiclab{} was introduced for showing the
  unpredictability introduced by social influence in cultural
  markets.} This paper generalizes our own results on the \musiclab{}
\cite{MusicLab14} and show that, even in this more general setting,
the expected profit can be optimized in polynomial time and social
influence is beneficial in maximizing expected profit.

\subsection{E-Commerce Models}

This paper also provides an alternative to traditional models from the
E-Commerce literature. In E-Commerce, the \emph{click-through rate}
(CTR) for a link $l$ is the probability that $l$ receives a
click. This probability may depend on a combination of factors, the
most significant ones being the relevance of the content and the
positioning of the links. The simplest model, which is pervasive in
the e-Commerce literature (e.g.,
\cite{kempe2008cascade,maille2012sponsored}), assumes that the CTRs
are independent. More precisely, this model assumes that the CTR of
link $l$ is the product of a position effect $\theta_k$ and a
relevance effect $q_l$. This simplification makes the model attractive
both from theoretical and practical standpoint.  For example, this
model is widely adopted in online advertising, since the optimal
allocation is simply obtained by sorting the advertisements by
decreasing $\theta_k q_i$. However, the independence assumption of
CTRs is not always justified. Experimental analysis using eye-tracking
\cite{joachims2005accurately,buscher2009you} has inspired cascade
models, first introduced in \cite{craswell2008experimental} and
subsequently generalized. Informally speaking, the cascade model
captures a sequential search, where users consider links from top to
bottom and only look at the next link if the previous link was not
selected. The model studied in this paper presents an alternative
which capture many interesting markets, while remaining tractable
computationally.

\subsection{Optimal Capacitated Multinomial Logit Assortment}

Rusmevichientong, Shen, and Shmoys, in their seminal paper
\cite{rusmevichientong2010dynamic}, consider a special case of the
CMLAPP where there is no visibility component, i.e., $\theta_i = 1$
$(i \in 1..c)$, and no social influence. In other words, they consider
the following problem.

\begin{definition}[Capacitated Multinomial Logit Assortment Problem] With
  the above notations, the Capacitated Multinomial Logit Assortment
  Problem (CMLAP) amounts to determining
\[
Z^*=\max\left\{ U(S) \mid S\subseteq {\cal P} \ \wedge \ \left|S\right|\leq c\right\}.
\]
where the profit of an assortment $S$ is defined as
$
U(S)=\sum_{i\in S}r_{i}P_{i}(S),
$
and the probability to select product $i \in S$ is given by
$
P_{i}(S)= \frac{u_{i}}{\sum_{j\in S}u_{j}+1}.
$
\end{definition}
\noindent
One of the main results of their paper is to show that the CMLAP can
be solved in polynomial time and we now review some of the key
concepts and intuition underlying their results.  The optimal profit
can be expressed as follows:
\[
\begin{array}{lll}
Z^* & = & \max\{ \lambda \in {\cal R} \mid \exists S \subseteq {\cal P}: |S| \leq c \ \wedge \ U(S) \geq \lambda \} \\
    & = & \max\{ \lambda \in {\cal R} \mid \exists S \subseteq {\cal P}: |S| \leq c \ \wedge \ \frac{\sum_{i \in S} u_{i}r_i}{\sum_{j\in S}u_{j}+1} \geq \lambda \} \\
    & = & \max\{ \lambda \in {\cal R} \mid \exists S \subseteq {\cal P}: |S| \leq c \ \wedge \  \sum_{i \in S} u_{i}(r_i- \lambda) \geq \lambda \}.
\end{array}
\]
Observe that, for a specific $\lambda$, it suffices to rank the
expressions $u_{i}(r_i- \lambda)$ $(i \in {\cal P})$ to find the
subset $S$ maximizing the term $\sum_{i \in X} u_{i}(r_i- \lambda)$.
We can then define a function $A:\mathbb{R}\rightarrow\left\{
  S \subseteq {\cal P} :\left|S\right|\leq c\right\}$
as
\[
A(\lambda)=\argmax_{S:\left|S\right|\leq c}\sum_{i\in S}u_{i}\left(r_{i}-\lambda\right)
\]
where ties are broken arbitrarily. The optimal profit can be shown to
be equivalent to
\[
Z^* = \max\{ U(A(\lambda)) \mid \lambda \in {\cal R} \}.
\]
It is valid to drop the condition $\sum_{i \in S} u_{i}(r_i- \lambda)
\geq \lambda$ due to dominance properties: When the condition is
violated for some $\lambda$, there exists a value $\lambda' < \lambda$
such that $U(A(\lambda)) < U(A(\lambda'))$. The key observation in
\cite{rusmevichientong2010dynamic} is that there are at most $O(N^2)$ values to
consider for $\lambda$ and hence at most $O(N^2)$ subsets of ${\cal
  P}$ to consider for finding the optimal assortment.

\section{Capacitated Multinomial Logit Assortment and Positioning}
\label{section:algorithm}

This section shows that the CMLAPP can be solved in polynomial time
building on the techniques used for the CMLAP. The key insight is that
the CMLAPP only needs to consider the candidate subsets of the CMLAP.
 
The first step in the proof consists in showing that there is no
benefit in introducing gaps in the positioning.

\begin{definition}[Gap-Free Position Assignment] Given an assortment $S$, a position assignment $\sigma$ is gap-free if
$
\{ \sigma_i \mid i \in S \} = 1..|S|.
$
\end{definition}

\begin{lemma}
\label{lemma:no-gap}
There is an optimal gap-free solution to any CMLAPP. 
\end{lemma}
\begin{proof}
  Let $\sigma$ be a position assignment with gaps. Assume that no product
  is assigned to position $k$ and let $l = \min\{i \in S \mid \sigma_i
  > k \}$ be the first product assigned to a position higher than $k$.
  Since $\theta_{k}\geq\theta_{\sigma_l}$, moving product $l$ from $\sigma_l$ to $k$ increases 
  its visibility. If this move is not profitable, it must be that 
\begin{align*}
 & \frac{\sum_{i\in S}\theta_{\sigma_i}u_{i}w_{i}+\left(\theta_{k}-\theta_{\sigma_l}\right)u_{l}r_{l}}{\sum_{j\in S}\theta_{\sigma_j}u_{j}+\left(\theta_{k}-\theta_{\sigma_l}\right)u_{l}+1}
\leq
\frac{\sum_{i\in S}\theta_{\sigma_i}u_{i}r_{i}}{\sum_{j\in S}\theta_{\sigma_j}u_{j}+1}. 
\end{align*}
Let $R=\sum_{i\in S}\theta_{\sigma_i}u_{i}r_{i}$ and $Q=\sum_{i\in
  S}\theta_{\sigma_i}u_{i}+1$. The above inequality becomes 
\begin{align*}{l}
 & \frac{R+\left(\theta_{k}-\theta_{\sigma_l}\right)u_{l}r_{l}}{Q+\left(\theta_{k}-\theta_{\sigma_l}\right)u_{l}}
\leq 
\frac{R}{Q} \\
 & QR+Q\left(\theta_{k}-\theta_{\sigma_l}\right)u_{l}r_{l}
\leq 
RQ+R\left(\theta_{k}-\theta_{\sigma_l}\right)u_{l} \\
& r_{l} \leq \frac{R}{Q}.
\end{align*}
The last inequality states that the marginal profit of product $l$ is
smaller or equal to the expected profit. We now show that we can
remove product $l$ while not degrading the profit. Indeed, $r_l \leq
\frac{R}{Q}$ implies
\begin{align*}
 & -R\theta u_{l}\leq-Q\theta u_{l}r_{l}  & (\cdot-\theta u_{l}) \\
 & RQ-R\theta u_{l}\leq RQ-Q\theta u_{l}r_{l} & (+RQ) \\
 & \frac{R}{Q}\leq\frac{R-\theta u_{l}r_{l}}{Q-\theta u_{l}}. 
\end{align*}
The result follows.
\end{proof}

\noindent
Lemma \ref{lemma:no-gap} indicates that the position assignment of an
assortment $S$ only need to consider positions $1..|S|$. Hence, the
position assignment is a bijection from $S$ to $1..|S|$. As in the
CMLAP, we define a function $B:\mathbb{R}\rightarrow\left\{ X\subseteq
  {\cal P} :\left|X\right|\leq c\right\}$ as
\[
B(\lambda)=\argmax_{S:\left|S\right|\leq c} \max_{\sigma: S \rightarrowtail 1..|S|} \sum_{i\in S} \theta_{\sigma_i} u_{i}\left(r_{i}-\lambda\right)
\]
which, given a value $\lambda$, specifies the assortment producing the
best profit for some optimal position assignment.  The optimal
position assignment $\sigma^{\lambda}_{S}$ for assortment $S$ and
value $\lambda$ is defined as
\[
\sigma^{\lambda}_{S} = \argmax_{\sigma: S \rightarrowtail 1..|S|} \sum_{i\in S} \theta_{\sigma_i} u_{i}\left(r_{i}-\lambda\right),
\]
Ties are broken arbitrarily in these two expressions. The optimal profit of the CMLAPP can be then reformulated as
\[
Z^* = \max\{ U(B(\lambda),\sigma^{\lambda}_{B(\lambda)}) \mid \lambda \in {\cal R} \}.
\]
The optimal position assignment $\sigma^{\lambda}_{S}$ can be computed
easily thanks to a rearrangement inequality.

\begin{lemma}[Rearrangement]
\label{lemma:rearrangment} Let $x_{1,}x_{2},\ldots,x_{n}$ and $y_{1,}y_{2},\ldots,y_{n}$
be real numbers (not necessarily positive) with 
$
x_{1}\leq x_{2}\leq\cdots\leq x_{n}\text{ and }y_{1}\leq y_{2}\leq\cdots\leq y_{n},
$
and let $\pi$ be any permutation of $\left\{ 1,2,\ldots,n\right\} $.
Then the following inequality holds:
\[
x_{1}y_{\pi_1}+x_{2}y_{\pi_2}+\cdots+x_{n}y_{\pi_n}\leq x_{1}y_{1}+x_{2}y_{2}+\cdots+x_{n}y_{n}.
\]
\end{lemma}
\begin{proof}
  We prove the inequality by induction on $n$. The statement is
  obvious for $n=1$. Suppose that it is true for $n=k-1$ and consider
  $n=k$. Let $m$ be an integer such that $\pi(m)=k$. Since
  $x_{k} \geq x_{m}$ and $y_{\pi(m)}=y_{k}\geq y_{\pi(k)}$, we have
\[
\left(x_{k}-x_{m}\right)\left(y_{k}-y_{\pi(k)}\right)\geq0
\]
This implies that 
\[
x_{k}y_{\pi(k)} + x_{m}y_{k} \leq x_{k}y_{k}+x_{m}y_{\pi(k)} \leq  x_{k}y_{k}+x_{m}y_{k} 
\]
and hence
\[
x_{k}y_{\pi(k)}\leq  x_{k}y_{k}.
\]
By induction hypothesis,
\[
x_{1}y_{\pi(1)}+\cdots+x_{m}y_{\pi(m)}+\cdots+x_{n-1}y_{\pi(n-1)}\leq x_{1}y_{1}+\cdots+x_{m}y_{m}+\cdots+x_{k-1}y_{k-1}.
\]
Combining these two inequalities proves the result.
\end{proof}

\noindent
The optimal position assignment must thus satisfy
\[
u_{\pi_1} (r_{\pi_1} - \lambda) \geq \ldots \geq u_{\pi_{|S|}} (r_{\pi_{|S|}} - \lambda)
\]
for some ranking assignment $\pi: 1..|S| \rightarrowtail S$ and the
optimal position assignment can be defined as the inverse of $\pi$. We
now show that the CMLAPP only needs to consider the assortments
considered by the CMLAP.

\begin{lemma}
\label{lemma:subset}
Let $\mathcal{B}=\left\{ B(\lambda):\lambda\in\mathbb{R}\right\}$ and
$\mathcal{A}=\left\{ A(\lambda):\lambda\in\mathbb{R}\right\} $. Then,
$\mathcal{B \subseteq A}$. 
\end{lemma}
\begin{proof}
  Consider $\lambda \in \mathcal{R}$, $S\in B(\lambda)$, and $\pi$ be
  the inverse of $\sigma^\lambda_S$. $\pi$ satisfies
\[
u_{\pi_1}\left(r_{\pi_1}-\lambda\right) \geq \cdots \geq u_{\pi_{|S|}}\left(r_{\pi_{|S|}}-\lambda\right).
\]
Consider first the case where $|S| < c$. Then, by optimality of
$B(\lambda)$ and $\theta_k \geq 0$, it must be the case that 
$
u_i (r_i - \lambda) \leq 0
$
for all $i \in {\cal P} \setminus S$. When $|S| = c$, it must be the case that
\[
u_i (r_i - \lambda) \geq u_j (r_j - \lambda) \;\;\;\;\;\; (i \in S \mbox{ and } j \in {\cal P} \setminus S)
\]
since otherwise swapping $i$ and $j$ would increase the profit.  As a
result, for any $S'\subseteq {\cal P}$ such that $|S'| \leq c$, we have that
\[
\sum_{i\in S}u_{i}\left(r_{i}-\lambda\right)\geq\sum_{j \in S'}u_{j}\left(r_{j}-\lambda\right).
\]
Hence $S \in \mathcal{A}$.
\end{proof}

\noindent
We are now in position to state that the CMLAPP can be solved in
polynomial time.
\begin{theorem}
\label{thm:poly}
CMLAPP can be solved in polynomial time.
\end{theorem}
\begin{proof}
  The result follows from Lemmas \ref{lemma:no-gap}--\ref{lemma:subset}
  and the fact that $\mathcal{A}$ can be computed in $O(N^2)$ time
  \cite{rusmevichientong2010dynamic}.
\end{proof}

This result is particularly interesting given the fact that the CMLAPP
violates the regularity assumption which states that the addition of
an option to a choice set should never increase the probability of
selecting an option in the original set.

\begin{definition}[Regularity Assumption]
  Let $X\subset Y \subseteq \mathcal{P}$. The regularity assumption
  for the CMLAP problem states that 
\[
 P_i(X) \geq P_i(Y) \;\;\; \forall i \in X \cup \{0\}.
\]
\end{definition}

\begin{lemma}
The CMLAPP problem violates the regularity assumption.
\end{lemma}
\begin{proof}
  Consider two products with $u_1 = 3, r_1=1 , u_2=1, r_2=\frac{5}{4}$
  and visibilities $\theta_1 =2, \theta_2 =1$. The optimal assortment
  for each value of $c$ is given by
\begin{center}
\begin{tabular}{|c|c|c|}
\hline 
$c$ & $1$ & $2$\tabularnewline
\hline 
\hline 
$S^{*}$ & $\left\{ 1\right\} $ & $\left\{ 2,1\right\} $\tabularnewline
\hline 
$Z^{*}$ & $0.8571$ & $0.9167$\tabularnewline
\hline 
$p_{0}$ & $0.1429$ & $0.1667$\tabularnewline
\hline 
\end{tabular}
\end{center}
where the optimal position assigment is specified by the order of the
products in assortment $S^*$. Although the
revenue $Z^*$ increases when moving from $\{1\}$ to $\{1,2\}$, the no-purchase option satisfies
\[
P_0(\{1\},\sigma_1) < P_1(\{1,2\},\sigma_2)
\]
for $\sigma_1(1) = 1$, $\sigma_2(1) = 2$, and $\sigma_2(2) = 1.$
\end{proof}

\section{The Benefits of Social Influence}
\label{section:social}

This section proves that, under social influence, the expected profit
increases over time when using the optimal assortment and position
assignment. This holds regardless of the number of products, their
utilities, and the visibilities. The derivation uses ranking
assignments, i.e., the inverse of position assignments, since they
make the notations and proofs simpler. Recall that $u_{i,t}$, the
utility of product $i$ at time $t$, is defined as the combination of the
inherent quality of the product $u_i$ and a non-decreasing and positive
social influence function $f(\cdot)$, i.e., $u_{i,t} = u_i +
f(d_{i,t}).$

In state $t$, the probability that product $i$ is selected given
assortment $S$ and ranking $\pi$ is
\[
  P_{i,t}(S,\pi)=\frac{\theta_{\pi_i}u_{i,t}}{\sum_{j\in S}\theta_{\pi_j}u_{j,t}+1}.
\]
Under social influence, the expected profit over time can be
considered as a Markov chain where state $t+1$ only depends on state
$t$. If product $k$ is selected at time $t$ and if assortment $S'$ and
ranking $\pi'$ are used at time $t+1$ and $i \in S'$, the probability
that product $i$ is selected at time $t+1$ is given by
\[
  \frac{\theta_{\pi'_i}u_{i,t+1}}{\sum_{j \in S': j \not=k}
  \theta_{\pi'_j}u_{j,t} + \theta_{\pi'_k} u_{k,t+1} + 1 }
\]
Since $f$ is non-decreasing, we define
$\epsilon_{i,t}=f(d_{i,t}+1)-f(d_{i,t}) \geq 0$ such that
\[
u_{i,t+1} = u_{i,t} +\epsilon_{i,t}.
\]
We denote by $U_{t+1}((S',\pi')|(S,\pi))$ the expected profit at time
$t+1$ conditional to time $t$ if assortment $S$ and ranking $\pi$ are
used at time $t$ and assortment $S'$ and ranking $\pi'$ are used at
time $t+1$. The core of the proof consists in showing that
\[
U_{t+1}((S^*,\pi^*)|(S^*,\pi^*)) \geq U_{t}(S^*,\pi^*),
\]
where $(S^*,\pi^*)$ is the optimal (assortment,ranking) pair at time
$t$. The expected profit at time $t+1$ conditional to time $t$ in 
this context is given by 
\begin{align*}
\sum_{j \in S^*}\left(P_{j,t}(S^*,\pi^*) 
\cdot\frac{\sum_{i \in S^*: i\not=j}\theta_{\pi^*_i}u_{i,t}r_{i}+\theta_{\pi^*_j}(u_{j,t}+\epsilon_{j,t})r_{j}}{\sum_{i \in S^*: i\not=j}\theta_{\pi^*_i}u_{i,t}+\theta_{\pi^*_j}(u_{j,t}+\epsilon_{j,t} )+1}\right) \\
+\left(1-\sum_{i \in S^*} P_{i,t}(S^*,\pi^*) \right)\cdot U_{t}(S^*.\pi^*).
\end{align*}
The bottom term captures the case where no product is selected at time
$t$. while the top term captures the cases where product $j$ is selected,
which increases its perceived utility for the next time step.

\begin{lemma} 
\label{lem:1}
Let $(S^*,\pi^*)$ be the optimal allocation in time step $t$. We have
$$U_{t+1}((S^*,\pi^*)|(S^*,\pi^*)) \geq U_{t}(S^*,\pi^*).$$
\end{lemma}
\begin{proof}
  Let $u=u_{1},\cdots,u_{n}$ represent the current state at time
  $t$. Without loss of generality, we can rename the songs so that
  $\pi^*_i = i$ and drop the ranking subscript $\pi^*$. We also omit
  writing the inclusion in $S^*$ in the summations and use
  $\epsilon_j$ to denote $\epsilon_{j,t}$.  The optimal expected
  profit at time $t$ can be written as
\[
U_t(S^*,\pi^*)=\frac{\sum_{i}\theta_{i}u_{i}r_{i}}{\sum_{i}\theta_{i}u_{i} +1}=\lambda^{*}.
\]
The expected profit $U_{t+1}((\pi^*,S^*)|(\pi^*,S^*))$ in time $t+1$
conditional to time $t$, which we denote by $U^*_{t+1}$, is given by
\[
U^*_{t+1} = \sum_{j}\left(\frac{\theta_{j}u_{j}}{\sum \theta_{i}u_{i} +1} \cdot 
\frac{\sum_{i\not=j}\theta_{i}u_{i}q_{i}+\theta_{j}(u_{j}+\epsilon_j)r_{j}}{\sum_{i\not=j}\theta_{i}u_{i}+\theta_{j}(u_{j}+\epsilon_j) +1 }\right)+
\left(1-\frac{\sum_{i}\theta_{i}u_{i}}{\sum_{i}\theta_{i}u_{i} +1}\right)\cdot
\frac{\sum_{i}\theta_{i}u_{i}r_{i}}{\sum_{i}\theta_{i}u_{i}+1 }
\]
\[
=\sum_{j}\left(\frac{\theta_{j}u_{j}}{\sum \theta_{i}u_{i}+1}\cdot\frac{\sum_{i}\theta_{i}u_{i}r_{i}+\epsilon_j \theta_{j}r_{j}}{\sum_{i}\theta_{i}u_{i}+\epsilon_j \theta_{j} +1 }\right)+
\left(1-\frac{\sum_{j}\theta_{j}u_{j}}{\sum_{i}\theta_{i}u_{i} +1}\right)\cdot\lambda^{*}.
\]
Proving
\begin{equation}
U^*_{t+1} \geq U_t(S^*,\pi^*) \label{eq:1}
\end{equation}
amounts to showing that 
\[
\sum_{j}\left(\frac{\theta_{j}u_{j}}{\sum \theta_{i}u_{i}+1}\cdot\frac{\sum_{i}\theta_{i}u_{i}r_{i}+\epsilon_j \theta_{j}r_{j}}{\sum_{i}\theta_{i}u_{i}+\epsilon_j \theta_{j}+ 1}\right)+
\left(1-\frac{\sum_{j}\theta_{j}u_{j}}{\sum_{i}\theta_{i}u_{i} +1}\right)\cdot\lambda^{*}
\geq \lambda^{*}.
\]
which reduces to proving
\[
\frac{1}{\sum_{i}\theta_{i}u_{i}+1}\sum_{j}\left[\frac{\theta_{j}^2u_{j}\epsilon_j }{\sum_{i}\theta_{i}u_{i}+\epsilon_j \theta_{j}+1}
\left( r_{j}-\lambda^{*}\right)\right] \geq 0
\]
or, equivalently, 
\[
\sum_{j}\left[\frac{\theta_{j}^2u_{j} \epsilon_j }{\sum_{i}\theta_{i}u_{i}+\epsilon_j\theta_{j}+1}\left( r_{j}-\lambda^{*}\right)\right] \geq 0.
\]
By definition of $B$, the optimality of $S^*$ implies that $S^* \in
B(\lambda^*)$. Hence, by optimality of $S^*$, $(r_i-\lambda^*)\geq 0$
for all $i\in S$.  The result follows since, for all $i$,
$\frac{\theta_{j}^2 u_{j} \epsilon_j
}{\sum_{i}\theta_{i}u_{i}+\epsilon_j\theta_{j}+1}\geq 0$.
\end{proof}

\noindent
We are now in position to state the main result of this section.

\begin{theorem}
\label{theorem:socialinfluence}
Under social influence, the expected profit is nondecreasing over
time when the optimal assortment and ranking are used at every step.
\end{theorem}
\begin{proof}
  Let $(S^*,\pi^*)$ and $(S^{**},\pi^{**})$ be the optimal solution
  in time step $t$ and $t+1$ respectively. Then,
\[
U_{t+1} ((\pi^{**}, S^{**})| (\pi^*,S^*))\geq U_{t+1} ((\pi^{*}, S^{*})| (\pi^*,S^*)) \geq U_{t} ((\pi^*,S^*)),
\]
where the first inequality follows from the optimality of
$(S^{**},\pi^{**})$ and the second inequality follows from Lemma
\ref{lem:1}.
\end{proof}

\noindent
Denote by $P^*$ the policy that applies the optimal CMLAPP solution at
each step. Theorem \ref{theorem:socialinfluence} entails an
interesting corollary about $P^*$. Indeed, when there is no social
influence, $P^*$ outperforms any policy $P$ since the optimization
problem is the same for every step. Since $P^*$ under social influence
outperforms (in expectation) $P^*$ with no social influence, it
follows that $P^*$ under social influence outperforms any policy not
using social influence.

\begin{corollary}
  \label{thm:increasing}
  In expectation, applying the optimal CMLAPP solution at each step
  with or without social influence outperforms any policy not using
  social influence.
\end{corollary}

\section{Conclusion}
\label{section:conclusion}

Motivated by practical applications in E-Commerce and recommendation
systems, this paper studied a multinomial logit model that captures
position bias, social influence, and limits on how many products can be
displayed. It showed how to optimize the expected profit for this
multinomial logit model in polynomial time. In addition, it showed
that it is beneficial to use social influence for maximizing the
expected profit. Experimental results on the \musiclab{}
\cite{MusicLab14}, a special case of the model with no capacity
constraint, an identity function $f$, an obligation to make a choice,
indicate that the algorithm provides significant improvements in
expected profit compared to simple ranking policies and that the
benefits of social influence can be substantial. Future work will
focus on generalizing the results to the case where market
participants have different preferences.

\bibliographystyle{spmpsci} 
\bibliography{model} 

\end{document}